\documentclass[11pt]{article}
\usepackage{geometry}
\usepackage{amsmath,amsthm,amssymb,mathtools}
\usepackage{xcolor}
\usepackage{booktabs,multirow,multicol,makecell}
\usepackage{enumerate}
\usepackage{graphicx}
\usepackage[pdftex,unicode, plainpages = false, pdfpagelabels, 
bookmarks=false,
bookmarksopen = true,
bookmarksnumbered = true,
breaklinks = true,
linktocpage,
pagebackref,
colorlinks = true,  
linkcolor = blue,
urlcolor  = blue,
citecolor = red,
anchorcolor = green,
hyperindex = true,
hyperfigures
]{hyperref} 
\hypersetup{
    colorlinks=true, 
    linktoc=all,     
    linkcolor=blue,
    citecolor=red,  
}
\usepackage{thmtools}   
\usepackage{cleveref}   

\newtheorem{theorem}{Theorem}
\newtheorem{lemma}{Lemma}
\newtheorem{definition}{Definition}
\newtheorem{corollary}{Corollary}
\newtheorem{hypothesis}{Hypothesis}

\newcommand{\bigOstar}[1]{\mathcal{O}^*(#1)}
\newcommand{\bigO}[1]{\mathcal{O}(#1)}

\DeclareMathOperator{\var}{var}
\DeclareMathOperator{\ltr}{lit}
\DeclareMathOperator{\WMC}{WMC}
\DeclarePairedDelimiter{\abs}{\lvert}{\rvert}

\newcommand{\nl}{\overline}
\newcommand{\formula}{\varphi}
\newcommand{\literal}{\ell}
\newcommand{\model}{\sigma}
\newcommand{\wfunc}{w}
\newcommand{\algcount}{\texttt{alg-count}}

\newcommand{\ind}[1]{\mathbf{1}[#1]}
\DeclareMathOperator{\Occ}{Occ}
\DeclareMathOperator{\OPT}{OPT}
\DeclareMathOperator{\supp}{supp}
\DeclareMathOperator{\score}{score}

\begin{document}
\title{Solving Restricted Versions of \#SAT Faster than $2^m$:\\ Algorithms and Lower Bounds
\footnotetext{Sanjay Jain and Frank Stephan were supported by Singapore Ministry of Education (MOE) AcRF Tier 2 grant MOE-000538-01. Additionally, Sanjay Jain was supported by NUS grant E-252-00-0021-09. Junqiang Peng was supported by the European Research Council (ERC) via grant NewPC, reference 101199930.}}
\author{
    Sanjay Jain\thanks{School of Computing, National University of Singapore, Singapore.}
    \\
    \texttt{sanjayjain@nus.edu.sg}
    \and
    Junqiang Peng\thanks{University of Bergen, Norway.} \thanks{University of Electronic Science and Technology of China, China.}
    \\
    \texttt{jqpeng0@foxmail.com}
    \and
    Frank Stephan\footnotemark[1] \thanks{Department of Mathematics, National University of Singapore, Singapore}
    \\
    \texttt{fstephan@nus.edu.sg}
    \and
    Mingyu Xiao\footnotemark[3]
    \\
    \texttt{myxiao@uestc.edu.cn}
}
\date{}
\maketitle

\begin{abstract}
Under the Strong Exponential Time Hypothesis (SETH), for any constant $\varepsilon > 0$, the \#SAT problem (counting the number of satisfying assignments of a CNF formula) cannot be solved in $\bigOstar{(2-\varepsilon)^n}$ or $\bigOstar{(2-\varepsilon)^m}$ time, where $n$ is the number of variables and $m$ is the number of clauses. On the other hand, simple exhaustive search and the standard inclusion-exclusion give $\bigOstar{2^n}$-time and $\bigOstar{2^m}$-time algorithms, respectively. A natural line of research is to beat these SETH-based barriers for restricted versions, such as \#$k$-SAT (where each clause contains at most $k$ literals) or \#$d$-occ-SAT (where each variable appears in at most $d$ clauses). It is known that the $2^n$-barrier can be broken when $k$ or $d$ is any constant; yet, whether the $2^m$-barrier can also be broken in these cases remains unknown.

We present two simple algorithms that overcome the $2^m$-barrier. First, for \#$k$-SAT, we give an $\bigOstar{2^{m(1-1/O(k))}}$-time algorithm, thereby breaking the $2^m$-barrier for any constant $k$. Second, for \#$2$-occ-SAT (which also includes the problem of counting the number of edge covers in a graph as a special case), we provide an algorithm with runtime $\bigOstar{2^{m\omega/3}}$, where $\omega < 2.3714$ is the matrix multiplication exponent. Furthermore, we demonstrate that both algorithms generalize to the weighted counting setting while preserving the same runtime bounds.

On the hardness side, we present a reduction from \#Max-$d$-SAT to \#$d$-occ-SAT, showing that for any $d\ge 3$, \#$d$-occ-SAT admits no $\bigOstar{(2-\varepsilon)^m}$-time algorithm for any constant $\varepsilon>0$, even on positive formulas, under the hypothesis that \#Max-3-SAT on $n$ variables admits no $\bigOstar{(2-\varepsilon)^n}$-time algorithm for any such $\varepsilon$.

\end{abstract}

\newpage
\section{Introduction}
The {\sc Satisfiability} (SAT) problem, which involves determining the satisfiability of a given propositional formula in Conjunctive Normal Form (CNF), is one of the most fundamental and important {\sf NP}-complete problems~\cite{conf/stoc/Cook71,journals/pit/levin73} in theoretical computer science. 
Its counting variant, {\sc Propositional Model Counting} (\#SAT), requires computing the exact number of satisfying assignments (also called models). As a canonical \#{\sf P}-complete~\cite{journals/tcs/Valiant79} problem, \#SAT finds applications in various areas (see \cite{series/faia/GomesSS21} for a comprehensive introduction). 

In this paper, we focus on exact exponential-time algorithms for \#SAT. This falls into the rich area of research devoted to designing exact algorithms for {\sf NP}-hard problems that are faster than na\"\i ve brute-force methods and improve runtimes beyond natural barriers. A recent survey by Nederlof~\cite{journals/csr/Nederlof26} refers to this field as the ``fine-grained complexity of {\sf NP}-complete problems''; we direct interested readers to this survey for a broader context.

By enumerating all $2^n$ possible assignments, SAT and \#SAT can be solved in $\bigOstar{2^n}$ time, where $\bigOstar{\cdot}$ omits polynomial factors in the input size and $n$ is the number of variables.
However, no algorithm for SAT with runtime $\bigOstar{(2-\varepsilon)^n}$ is known for any constant $\varepsilon > 0$. Under the Strong Exponential Time Hypothesis (SETH)~\cite{journals/jcss/ImpagliazzoP01}, no such algorithm exists. 

While the $2^n$-barrier is insurmountable under SETH for general SAT and \#SAT, it can be broken for \emph{restricted} classes of formulas.
Notable examples include $k$-CNF (where each clause contains at most $k$ literals) and $d$-occ-CNF (where each variable appears in at most $d$ clauses). The corresponding problems are denoted by $k$-SAT, $d$-occ-SAT, \#$k$-SAT, and \#$d$-occ-SAT, respectively.
It is well-established that $k$-SAT can be solved in $\bigOstar{2^{n(1-\Omega(1/k))}}$ time via various algorithmic paradigms~\cite{journals/jacm/PaturiPSZ05,conf/focs/Schoning99,conf/soda/ChanW16,conf/stoc/BrakensiekG19}, and $d$-occ-SAT can be solved in $\bigOstar{2^{n(1-\Omega(1/\log d))}}$ time~\cite{conf/coco/CalabroIP06,series/faia/DantsinH21}.
Remarkably, the same runtime bounds can be achieved for the counting counterparts: Impagliazzo, Matthews, and Paturi~\cite{conf/soda/ImpagliazzoMP12} gave Las Vegas randomized algorithms for \#$k$-SAT in $\bigOstar{2^{n(1-\Omega(1/k))}}$ time and \#$d$-occ-SAT in $\bigOstar{2^{n(1-\Omega(1/\log d))}}$ time. Subsequently, Chan and Williams~\cite{conf/soda/ChanW16} developed the deterministic algorithms achieving the respective runtime bounds, albeit requiring exponential space.

In this paper, we turn our attention to another natural parameter: the number of clauses $m$.
The complexity landscape is markedly different under this parameterization. In particular, SAT can be solved in $\bigOstar{1.2226^m}$ time~\cite{journals/tcs/ChuXZ21}, whereas \#SAT can be solved in $\bigOstar{2^m}$ time via standard inclusion-exclusion~\cite{journals/siamcomp/Iwama89,journals/ipl/Lozinskii92}. Moreover, the authors in \cite{journals/talg/CyganDLMNOPSW16} showed that no $\bigOstar{(2-\varepsilon)^m}$-time algorithm exists for any $\varepsilon>0$ under SETH.

For restricted versions, it is known that weighted \#2-SAT can be solved in $\bigOstar{1.1082^m}$ time and exponential space~\cite{conf/ijcai/0001S025} (the definition of the weighted setting is deferred later). Furthermore, an $\bigOstar{1.1082^L}$-time algorithm for general \#SAT can be obtained by reducing an instance of length $L$ (the sum of the lengths of all clauses) to a weighted \#2-SAT instance with $m = L$ clauses~\cite{conf/lics/BannachDGH25,conf/sosa/BannachDGH26} and invoking the aforementioned algorithm~\cite{conf/ijcai/0001S025}. For \#$k$-SAT, since $L \le km$, this yields an $\bigOstar{1.1082^{km}}$-time algorithm and breaks the $2^m$-barrier for \#$k$-SAT when $k \le 6$. Nevertheless, whether this barrier can be broken for any arbitrary fixed $k$ remains unknown. Moreover, for \#$d$-occ-SAT, to the best of our knowledge, no algorithm was known to break the $2^m$-barrier even for the most restrictive case of $d=2$.

\begin{table}[t]
    \centering
    \begin{tabular}{llll}
    \toprule
    \multirow{2}{*}{\textbf{Problems}} & \multicolumn{2}{c}{\textbf{Previous Results}} & \multirow{2}{*}{\textbf{Our Results}} \\
    \cmidrule(lr){2-3}
    & \textbf{Parameter $n$} & \textbf{Parameter $m$} & \\
    \midrule
    \multirow{2}{*}{\#SAT} & $\bigOstar{2^n}$ & $\bigOstar{2^m}$ & \multirow{2}{*}{--} \\
    & {\scriptsize Det. \& Poly-S} & {\scriptsize Det. \& Poly-S \cite{journals/siamcomp/Iwama89,journals/ipl/Lozinskii92}} &\\
    \midrule
    \multirow{2}{*}{\#$2$-SAT} & $\bigOstar{1.2377^n}$ & $\bigOstar{1.1082^m}$ & \multirow{2}{*}{--} \\
    & {\scriptsize Det. \& Poly-S \cite{conf/iwpec/Wahlstrom08}} & {\scriptsize Det. \& Exp-S \cite{conf/ijcai/0001S025}} & \\
    \midrule
    \multirow{3}{*}{{\makecell[l]{\#$k$-SAT \\($k\ge 3$)}}} & $\bigOstar{2^{n(1-\Omega(1/k))}}$ & {$\bigOstar{1.1082^{km}}$ for $k\le 6$} & \multirow{2}{*}{$\bigOstar{2^{m(1-\Omega(1/k))}}$} \\
    & {\scriptsize Rand. \& Poly-S \cite{conf/soda/ImpagliazzoMP12}} & {\scriptsize \makecell[l]{Det. \& Exp-S \\ \cite{conf/lics/BannachDGH25,conf/sosa/BannachDGH26,conf/ijcai/0001S025}}} & \multirow{2}{*}{\scriptsize Det. \& Poly-S [\Cref{sec:k-CNF}]} \\
    & {\scriptsize Det. \& Exp-S \cite{conf/soda/ChanW16}} & $\bigOstar{2^m}$ for $k\ge 7$ &\\
    \midrule
    \multirow{2}{*}{\#$2$-occ-SAT} & $\bigOstar{1.2282^n}$ & \multirow{2}{*}{$\bigOstar{2^m}$} & $\bigOstar{2^{m\omega/3}}\subseteq \bigOstar{1.7297^m}$ \\
    & {\scriptsize \makecell[l]{Det. \& Exp-S \\\cite{conf/lics/BannachDGH25,conf/sosa/BannachDGH26,conf/ijcai/0001S025}}} & & {\scriptsize Det. \& Exp-S [\Cref{sec:2-occ-CNF}]} \\
    \midrule
    \multirow{3}{*}{\makecell[l]{\#$d$-occ-SAT \\($d\ge 3$)}} & $\bigOstar{2^{n(1-\Omega(1/\log d))}}$ & \multirow{3}{*}{$\bigOstar{2^m}$} & \multirow{3}{*}{\makecell[l]{No $\bigOstar{(2-\varepsilon)^m}$-time algorithm\\ under the \#Max-3-SAT Hypothesis \\  {\scriptsize[\Cref{sec:max-to-occ}]}}} \\
    & {\scriptsize Rand. \& Poly-S \cite{conf/soda/ImpagliazzoMP12}} &&\\
    & {\scriptsize Det. \& Exp-S \cite{conf/soda/ChanW16}} &&\\
    \bottomrule
\end{tabular}

    \caption{Summary of known results and our results. The runtime bound $\bigOstar{1.2282^n}$ for \#$2$-occ-SAT follows directly from the $\bigOstar{1.1082^L}$-time algorithm for \#SAT since $L \le 2n$. Here, $\omega<2.3714$ is the matrix multiplication exponent~\cite{conf/soda/AlmanDWXXZ25}. `Det.' and `Rand.' refer to `deterministic' and `randomized', respectively, while `Poly-S' and `Exp-S' stand for `polynomial space' and `exponential space'.}
    \label{tab:intro-summary}
\end{table}

\paragraph*{Why parameter $m$.}
The parameters $n$ and $m$ capture two different natural search spaces: assignments of variables and subsets of clauses, underlying the $\bigOstar{2^n}$-time exhaustive-search and $\bigOstar{2^m}$-time inclusion-exclusion algorithms, respectively.
Results parameterized by $n$ do not directly yield savings over the $2^m$-barrier; for example, although a $k$-CNF satisfies $n\le km$, this relation turns an $\bigOstar{2^{n(1-\Omega(1/k))}}$-time algorithm only into an $\bigOstar{2^{m(k-\Omega(1))}}$-time algorithm. Thus, parameter $m$ provides a natural perspective complementary to parameter $n$.

Furthermore, \#SAT parameterized by $m$ contains \#{\sc All Set Cover} parameterized by the universe size as a special case.
Given a universe $\mathcal U$ and subsets $S_1,\ldots,S_r$ of $\mathcal U$ (not necessarily distinct), \#{\sc All Set Cover} asks for the number of index sets $I\subseteq\{1,\ldots,r\}$ such that $\bigcup_{i\in I}S_i=\mathcal U$.
Introducing one variable $x_i$ for each $i\in\{1,\ldots,r\}$ and one positive clause $C_u=\bigvee_{i:\,u\in S_i}x_i$ for each $u\in\mathcal U$ gives a bijection between satisfying assignments and set covers, with $m=|\mathcal U|$.

\paragraph*{Our main results.}
In this work, we present algorithms that break the $2^m$-barrier for restricted versions of \#SAT, namely \#$k$-SAT for any fixed $k$ and \#$2$-occ-SAT, as well as hardness results for \#$d$-occ-SAT for every fixed $d\ge 3$ via a reduction from \#Max-$d$-SAT. Known results and our results for \#SAT and its restricted versions are summarized in \Cref{tab:intro-summary}.

Our algorithmic results are as follows.

\begin{restatable}{theorem}{kCnfComplexityThm}
\label{thm:k-cnf-complexity}
\#$k$-SAT with $m$ clauses can be solved in $\bigOstar{2^{m(1 - \frac{1}{5k})}} \subseteq \bigOstar{2^{m(1 - \Omega(1/k))}}$ time and polynomial space.
\end{restatable}

\begin{restatable}{theorem}{TwoOccCNFComplexityThm}\label{thm:2-occ-cnf-complexity}
    \#$2$-occ-SAT with $m$ clauses can be solved in $\bigOstar{2^{m \omega / 3}}\subseteq \bigOstar{1.7297^m}$ time and exponential space, where $\omega < 2.3714$ denotes the matrix multiplication exponent.
\end{restatable}

An \emph{edge cover} of a multigraph $G=(V,E)$ is a subset of edges incident to every vertex in $V$, where parallel edges are treated as distinct and a loop covers its incident vertex.
{\sc \#All Edge Cover} involves computing the number of edge covers of a given multigraph (where parallel edges and loops are allowed), which naturally reduces to \#$2$-occ-SAT.
Specifically, {\sc \#All Edge Cover} is the special case of {\sc \#All Set Cover} in which the universe is $V$ and each edge corresponds to the set of its incident vertices, of size one or two.
Under the reduction to \#SAT described above, each non-loop edge variable appears in the two clauses corresponding to its endpoints, while each loop variable appears only in the clause corresponding to its incident vertex.
Thus each variable occurs in at most two clauses, yielding a \#$2$-occ-SAT instance with $m=|V|$ clauses.
To the best of our knowledge, no algorithm for {\sc \#All Edge Cover} with runtime $\bigOstar{(2-\varepsilon)^{|V|}}$ for any $\varepsilon>0$ was known.
Our \Cref{thm:2-occ-cnf-complexity} implies an algorithm for {\sc \#All Edge Cover}:

\begin{restatable}{corollary}{AllEdgeCoverCor}
\label{cor:edge-cover}
{\sc \#All Edge Cover} on multigraphs with $n$ vertices can be solve in $\bigOstar{2^{n \omega / 3}}\subseteq \bigOstar{1.7297^{n}}$ time and exponential space.
\end{restatable}


We further demonstrate that our algorithms (\Cref{thm:k-cnf-complexity} and \Cref{thm:2-occ-cnf-complexity}) generalize to the weighted counting setting while preserving the same runtime bounds.
In the weighted setting, each literal is associated with a weight; the weight of a truth assignment is the product of the weights of the literals it satisfies, and the weighted model count of a formula is the sum of these weights over all satisfying assignments.

Our hardness results are based on the following hypothesis.
\begin{hypothesis}[\#Max-3-SAT Hypothesis]\label{hyp:max-3-sat}
For every constant $\varepsilon>0$, \#Max-3-SAT on $n$ variables cannot be solved in $\bigOstar{(2-\varepsilon)^n}$ time.
\end{hypothesis}
Here, \#Max-$d$-SAT asks for the number of truth assignments that maximize the number of satisfied clauses in a CNF formula whose clauses contain at most $d$ literals.
We note that the stronger hypothesis that Max-3-SAT on $n$ variables cannot be solved in $\bigOstar{(2-\varepsilon)^n}$ time for any constant $\varepsilon>0$ has been used in the literature~\cite{conf/soda/LincolnWW18,conf/coco/BringmannFK19,conf/esa/JansenW20,conf/itcs/AlmanW20}.

We show that, for every fixed $d\ge3$ and constant $\varepsilon>0$, \#$d$-occ-SAT with $m$ clauses cannot be solved in $\bigOstar{(2-\varepsilon)^m}$ time, even on positive formulas, under the \#Max-3-SAT Hypothesis.
A formula is positive if it contains only positive literals.
To this end, we prove the following theorem.

\begin{restatable}{theorem}{MaxToOccReductionThm}\label{thm:max-to-occ}
For every fixed $d\ge2$, there is a polynomial-time non-adaptive Turing reduction from \#Max-$d$-SAT on $n$ variables to \#$d$-occ-SAT on positive formulas with exactly $n$ clauses.
Thus, for any constant $\gamma\ge1$, an $\bigOstar{\gamma^m}$-time algorithm for the latter problem with $m$ clauses yields an $\bigOstar{\gamma^n}$-time algorithm for the former problem, preserving polynomial space.
\end{restatable}

We note that \#$d$-occ-SAT on positive formulas parameterized by the number of clauses is equivalent to \#{\sc All Set Cover} with sets of size at most $d$ parameterized by the number of elements.
Here, the sets need not be distinct.
The reduction from \#{\sc All Set Cover} to \#SAT described above produces positive formulas; when each set has size at most $d$, each variable occurs in at most $d$ clauses.
Conversely, \#$d$-occ-SAT on positive formulas also naturally reduces to \#{\sc All Set Cover} with sets of size at most $d$: each clause corresponds to one element in the universe, and each variable corresponds to a subset of the universe containing precisely the elements corresponding to clauses containing that variable.
Each such subset has size at most $d$, and selecting the sets corresponding to variables assigned true gives a bijection between satisfying assignments and set covers.
Based on this equivalence, we have the following consequences of \Cref{thm:max-to-occ}.
\begin{corollary}\label{cor:edge-cover-to-max}
For every constant $\varepsilon>0$, an algorithm for \#{\sc All Edge Cover} on multigraphs with $N$ vertices running in $\bigOstar{(2-\varepsilon)^N}$ time and polynomial space yields an algorithm for \#Max-2-SAT on $n$ variables running in $\bigOstar{(2-\varepsilon)^n}$ time and polynomial space.
\end{corollary}
\begin{corollary}\label{cor:occ-conditional-lower-bound}
For every fixed $d\ge3$, \#{\sc All Set Cover} on a universe of $n$ elements with sets of size at most $d$ cannot be solved in $\bigOstar{(2-\varepsilon)^n}$ time for any constant $\varepsilon>0$, under the \#Max-3-SAT Hypothesis.
\end{corollary}

\section{Preliminaries}

Let $V$ be a set of \emph{Boolean variables}. A Boolean variable $x \in V$ can be assigned a value of $1$ (\textsc{True}) or $0$ (\textsc{False}). A variable $x$ has two corresponding \emph{literals}: the positive literal $x$ and the negative literal $\nl{x}$. We use $\nl{\literal}$ to denote the negation of literal $\literal$, so $\nl{\nl{\literal}}=\literal$.
A \emph{clause} on $V$ is a set of literals on $V$, which might be empty. A \emph{CNF formula} $\formula$ over $V$ is a multiset of clauses on $V$, allowing identical clauses. Different copies of the same clause are treated as distinct.
Let $\var(\formula)$ denote the variable set of $\formula$. For a collection $S$ of clauses, we use $\ltr(S)$ to denote the set of literals appearing in the clauses of $S$ (i.e., $\literal \in \ltr(S)$ if $\literal \in C$ for some $C\in S$).

An \emph{assignment} for a variable set $V$ is a mapping $\model: V\rightarrow\{0, 1\}$. Given an assignment $\model$, a clause is \emph{satisfied} by $\model$ if at least one literal in it evaluates to $1$ under $\model$. An assignment for $\var(\formula)$ is called a \emph{satisfying assignment} (or \emph{model}) of $\formula$ if $\model$ satisfies all clauses in $\formula$. We denote by $\mathcal{M}(\formula)$ the set of all satisfying assignments of $\formula$, and let $\#(\formula) := |\mathcal{M}(\formula)|$.

Without loss of generality, we assume throughout that no clause contains both a variable and its negation. Such tautological clauses can be deleted while retaining the variable set, without changing
the set of satisfying assignments.
The \emph{degree} of a variable $x$, denoted by $\deg_{\formula}(x)$, is the number of its occurrences in $\formula$. A formula $\formula$ is called \emph{$k$-CNF} if each clause has length at most $k$, and called \emph{$d$-occ-CNF} if each variable appears in at most $d$ clauses.

Throughout, we use $\omega$ to denote the \emph{matrix multiplication exponent}, that is, a constant such that two $N\times N$ matrices can be multiplied using $\bigO{N^\omega}$ arithmetic operations. Known matrix multiplication algorithms allow us to choose $\omega<2.3714$~\cite{conf/soda/AlmanDWXXZ25}.

\section{An Algorithm for \#$k$-SAT via Clause Branching}\label{sec:k-CNF}
In this section, we present an $\bigOstar{2^{m(1-\Omega(1/k))}}$-time branching algorithm for \#$k$-SAT, where $m$ is the number of clauses. We first introduce some necessary notation.

For a literal $\literal$, let $\formula[\literal=1]$ denote the formula obtained by assigning $\literal = 1$ (removing clauses containing $\literal$ and deleting $\nl{\literal}$ from others), and let $\formula[\literal=0]:=\formula[\nl{\literal}=1]$. For a clause $C=(\literal_1\lor \cdots \lor \literal_k)$, we write $\formula[C=0] := \formula[\literal_1=0, \dots, \literal_k=0]$ to denote the formula obtained by falsifying $C$, and let $\formula \setminus \{C\}$ denote the formula obtained by removing (one copy of) $C$ from $\formula$.
Clause deletion retains the variable set, while assignment removes the assigned variables and retains all other variables.

The algorithm consists of only a simple clause-based branching rule and two reduction rules. We will first introduce these rules and then present the algorithm along with the analysis. 

\subsection{Clause Branching and Isolated Clause Reduction}

\begin{lemma}[clause branching]\label{lem:clause-branching}
    Let $C$ be a clause in formula $\formula$. Then
    \[
        \#(\formula) = \#(\formula\setminus \{C\}) - \#((\formula\setminus\{C\})[C=0]).
    \]
\end{lemma}
\begin{proof}[Proof of \Cref{lem:clause-branching}]
    We partition the satisfying assignments of formula $\formula \setminus \{C\}$ into those satisfying $C$ (which are exactly the satisfying assignments of formula $\formula$) and those falsifying $C$. Since falsifying $C$ requires assigning $0$ to all its literals, the number of such assignments is exactly $\#((\formula \setminus \{C\})[C=0])$. Thus, we have $\#(\formula \setminus \{C\}) = \#(\formula) + \#((\formula \setminus \{C\})[C=0])$. Rearranging this equation yields the lemma.
\end{proof}

A clause $C$ is called \emph{isolated} if all variables in $C$ only appear once in the formula, i.e., $\deg_{\formula}(x)=1$ for all $x\in \var(C)$. Since it shares no variables with the rest of the formula, we can compute its satisfying assignments independently, leading to the following reduction rule.

\begin{lemma}[isolated clause reduction]\label{lem:isolated-clause}
    Let $C$ be an isolated clause in formula $\formula$. Let $\operatorname{rem}(\formula,C)$ denote the formula obtained by deleting $C$ and removing the variables in $\var(C)$ from the variable set. Then,
    \[
        \#(\formula) = \#(\operatorname{rem}(\formula,C))\cdot(2^{|C|} - 1).
    \]
\end{lemma}
\begin{proof}[Proof of \Cref{lem:isolated-clause}]
    Since $C$ is isolated, its variables do not appear in any other clause. Thus, the satisfying assignments of $\formula$ are the
    Cartesian product of the satisfying assignments of $C$ and those of $\operatorname{rem}(\formula,C)$.
    Exactly one of the $2^{|C|}$ assignments to $\var(C)$ falsifies $C$,
    so $C$ has $2^{|C|}-1$ satisfying assignments.
    This yields the claimed equality.
\end{proof}

\subsection{The Algorithm}\label{subsec:branching}
Let $V_{\ge 2}(\formula, C) := \{ x \in \var(C) \mid \deg_{\formula}(x) \ge 2 \}$. If $\formula$ contains no tautological clauses, these are precisely the variables in $C$ that also appear in at least one other clause in $\formula \setminus \{C\}$.
The algorithm, denoted as $\algcount(\formula)$, recursively computes $\#(\formula)$ by applying the following steps in decreasing order of priority. 

\begin{enumerate}[\textbf{Step \arabic*.}]
    \item \textbf{Base Cases:} 
    If $\formula$ contains the empty clause, return $0$. If $\formula = \emptyset$, return $2^{\abs{\var(\formula)}}$.

    \item \textbf{Tautological Clause Reduction:} If there exists a clause $C$ containing both $x$ and $\nl{x}$ for some variable $x$, return $\algcount(\formula\setminus\{C\})$, retaining the variable set.

    \item \textbf{Isolated Clause Reduction:} 
    If there exists an isolated clause $C \in \formula$ (i.e., $V_{\ge 2}(\formula, C) = \emptyset$), apply \Cref{lem:isolated-clause} and return $(2^{|C|} - 1) \cdot \algcount(\operatorname{rem}(\formula,C))$.
    
    \item \textbf{Clause Branching:} 
    Select a \emph{pivot clause} $C \in \formula$ that minimizes $|V_{\ge 2}(\formula, C)|$. Since Steps 2 and 3 were not triggered, we are guaranteed that $1 \le |V_{\ge 2}(\formula, C)| \le k$. Pick an arbitrary variable $x \in V_{\ge 2}(\formula, C)$, and choose a \emph{neighboring clause} $C' \in \formula \setminus \{C\}$ such that $x \in \var(C')$. Apply the clause branching rule (\Cref{lem:clause-branching}) on $C'$ and return
    $
        \algcount(\formula \setminus \{C'\}) - \algcount((\formula \setminus \{C'\})[C'=0])
    $.
\end{enumerate}

To bound the size of the search tree generated by $\algcount(\formula)$, we adopt the Measure and Conquer paradigm~\cite{journals/jacm/FominGK09}.
Let $m(\formula)$ denote the number of clauses in $\formula$.
The measure of a formula $\formula$ is defined as:
\[
    \mu(\formula) := m(\formula) + \frac{\eta(\formula)}{k}, \text{ where } \eta(\formula) := \min_{C \in \formula} |V_{\ge 2}(\formula, C)|
\]
with the convention that $\eta(\emptyset)=0$.
Intuitively, the fractional term $\eta(\formula)/k$ tracks the progress toward isolating a clause. Since $\eta(\formula) \le k$, we have $\mu(\formula) \le m(\formula) + 1$. Thus, any bound of $\bigOstar{c^\mu}$ for some $c$ implies $\bigOstar{c^m}$.

\begin{lemma}\label{lem:reduction-measure}
    Applying tautological clause reduction or isolated clause reduction does not increase the measure $\mu(\formula)$.
\end{lemma}
\begin{proof}[Proof of \Cref{lem:reduction-measure}]
    For tautological clause reduction, let $\formula'=\formula\setminus\{C\}$.
    Since $m(\formula')=m(\formula)-1$ and $\eta(\formula')\le k$,
    we have
    $
        \mu(\formula')\le m(\formula)\le\mu(\formula)
    $.
    
    For an isolated clause $C \in \formula$, $|V_{\ge 2}(\formula, C)| = 0$. By the definition of $\eta(\formula)$, we have $\eta(\formula) = 0$ and $\mu(\formula) = m(\formula)$. The reduction removes $C$ and $\var(C)$, yielding a residual formula $\formula' := \operatorname{rem}(\formula,C)$ with $m(\formula') = m(\formula) - 1$. Although the new minimum $\eta(\formula')$ might increase, it is bounded by $k$. Thus, the new measure is:
    \[
        \mu(\formula') = m(\formula') + \frac{\eta(\formula')}{k} \le m(\formula) - 1 + \frac{k}{k} = m(\formula) = \mu(\formula).
    \]
    Therefore, the measure never increases during this reduction.
\end{proof}

Next, we analyze how the measure $\mu(\formula)$ decreases in each recursive step. The key observation is that our pivot clause $C$ acts as a ``witness'' to prevent the fractional part of the measure from increasing after removing its neighboring clause $C'$, and even forces a decrease in the fractional part when falsifying $C'$.

\begin{lemma}\label{lem:measure-drop}
    In Step 4 of the algorithm, branching on $C'$ reduces the measure $\mu$ by at least $\Delta_1 \ge 1$ in the branch $\formula \setminus \{C'\}$ and by at least $\Delta_2 \ge 1 + 1/k$ in the branch $(\formula \setminus \{C'\})[C'=0]$.
\end{lemma}
\begin{proof}[Proof of \Cref{lem:measure-drop}]
    Let $C$ be the pivot clause chosen in Step 4, with $|V_{\ge 2}(\formula, C)| = \eta(\formula)\ge 1$.
    \begin{itemize}
        \item \textbf{Branch 1 (removing $C'$):} The formula becomes $\formula' := \formula \setminus \{C'\}$. The number of clauses decreases by 1, so $m(\formula') = m(\formula) - 1$. Since $C' \neq C$, the pivot clause $C$ remains in $\formula'$. Removing $C'$ can only decrease or maintain the degrees of the remaining variables, so  $\eta(\formula') \le \eta(\formula)$. The total measure drop is:
        \[
            \Delta_1 := \left(m(\formula) + \frac{\eta(\formula)}{k}\right) - \left(m(\formula') + \frac{\eta(\formula')}{k}\right) = 1 + \frac{\eta(\formula) - \eta(\formula')}{k} \ge 1.
        \]
        
        \item \textbf{Branch 2 (falsifying $C'$):} The formula becomes $\formula'' := (\formula \setminus \{C'\})[C'=0]$. The clause count $m$ decreases by at least 1. Furthermore, since $x \in \var(C) \cap \var(C')$, falsifying $C'$ assigns a truth value to $x$. This removes $x$ from $C$, which implies $|V_{\ge 2}(\formula'', C)| \le |V_{\ge 2}(\formula, C)| - 1$. Thus, $\eta(\formula'') \le \eta(\formula) - 1$ and the measure drop is:
        \[
            \Delta_2 := \left(m(\formula) + \frac{\eta(\formula)}{k}\right) - \left(m(\formula'') + \frac{\eta(\formula'')}{k}\right) \ge 1 + \frac{\eta(\formula) - \eta(\formula'')}{k} \ge 1 + \frac{1}{k}.
        \]
    \end{itemize}
    Note that if $C$ happens to be satisfied and hence removed as a side effect in the second branch, $m$ decreases by at least 2. Even if $\eta(\formula)$  could increase from its minimum possible value of $1$ to the maximum possible value of $k$, the measure drop is at least $2 + \frac{1-k}{k} = 1 + \frac{1}{k}$. Thus, the claimed bounds hold in all cases.
\end{proof}

With \Cref{lem:reduction-measure} and \Cref{lem:measure-drop}, we are ready to prove \Cref{thm:k-cnf-complexity}.

\kCnfComplexityThm*
\begin{proof}[Proof of \Cref{thm:k-cnf-complexity}]
    Tautological clause reduction preserves the set of satisfying assignments. Together with \Cref{lem:isolated-clause} and \Cref{lem:clause-branching}, the correctness of $\algcount(\formula)$ is guaranteed.
    Clearly, the algorithm takes polynomial space.
    Let $T(\mu)$ denote the maximum number of leaves in the search tree of $\algcount(\formula)$ for a formula $\formula$ of measure at most $\mu$.
    Both reduction rules can be performed in polynomial time, and they do not increase the measure by \Cref{lem:reduction-measure}.
    By \Cref{lem:measure-drop}, the two subinstances created by clause branching have measure decreases of at least $1$ and $1+1/k$, respectively. Hence, at branching nodes, we have the following recurrence:
    \[
        T(\mu) \le T(\mu - 1) + T(\mu - 1 - 1/k).
    \]
    Since every measure is an integer multiple of $1/k$, we prove by induction on $k\mu$ that $T(\mu) \le 2^{\mu(1-\epsilon)}$ for $\epsilon := \frac{1}{5k}$. Note that this value is chosen for the simplicity of the proof, rather than optimizing the bound; we believe that the constant factor in $\epsilon$ can be further optimized via a more careful analysis.
    The bound trivially holds for the base cases.
    Assuming the bound holds for all values smaller than $\mu$, substituting it into the recurrence gives:
    \[
        T(\mu) \le 2^{(\mu - 1)(1 - \frac{1}{5k})} + 2^{(\mu - 1 - \frac{1}{k})(1 - \frac{1}{5k})}.
    \]
    To complete the inductive step, the right-hand side must be at most $2^{\mu(1 - \frac{1}{5k})}$. Dividing this inequality by $2^{\mu(1 - \frac{1}{5k})}$ yields the condition:
    \[
        2^{-(1 - \frac{1}{5k})} + 2^{-(1 + \frac{1}{k})(1 - \frac{1}{5k})} \le 1.
    \]
    Multiplying both sides by $2$, we obtain the equivalent condition:
    \[
        2^{\frac{1}{5k}} + 2^{\frac{1}{5k} - \frac{1}{k}(1 - \frac{1}{5k})} \le 2  \iff 2^{\frac{1}{5k}} \left( 1 + 2^{-\frac{5k - 1}{5k^2}} \right) \le 2.
    \]
    To verify this inequality, we apply the bounds $2^y \le 1 + y$ and $1 + 2^{-y} \le 2 - y/2$, which hold for $y \in [0, 1]$. Since $k \ge 1$, the exponents $y_1 := \frac{1}{5k}$ and $y_2 := \frac{5k - 1}{5k^2}$ are in $(0, 1]$. Applying these bounds gives:
    \begin{align*}
        2^{\frac{1}{5k}} \left( 1 + 2^{-\frac{5k - 1}{5k^2}} \right) 
        &\le \left( 1 + \frac{1}{5k} \right) \left( 2 - \frac{5k - 1}{10k^2} \right) \\
        &= 2 - \frac{5k - 1}{10k^2} + \frac{2}{5k} - \frac{5k - 1}{50k^3} \\
        &= 2 - \frac{k - 1}{10k^2} - \frac{5k - 1}{50k^3} < 2 \quad \text{ (for any integer $k\ge 1$).}
    \end{align*}
    This completes the induction. The total running time is bounded by $O^*(2^{\mu(1 - \frac{1}{5k})})$, which establishes the claimed $O^*(2^{m(1 - \frac{1}{5k})})$ bound since $\mu(\formula) \le m(\formula) + 1$.
\end{proof}

\section{An Algorithm for \#$2$-occ-SAT via Split-and-List}\label{sec:2-occ-CNF}
In this section, we present an $\bigOstar{2^{m\omega /3}}$-time algorithm for \#2-occ-SAT, where $m$ is the number of clauses. 
Our approach builds upon the standard inclusion-exclusion formulation for \#SAT~\cite{journals/siamcomp/Iwama89,journals/ipl/Lozinskii92}. One can evaluate this formulation by enumerating all $2^m$ subsets of clauses, taking $\bigOstar{2^m}$ time. To avoid this, we employ the \emph{Split-and-List} method~\cite{journals/tcs/Williams05,journals/ipl/Koivisto06}. For convenience and clarity, we will first recast the subset-based inclusion-exclusion formulation into an equivalent sum-product form, and then apply Split-and-List to evaluate it. In a sum-product problem, the objective is to evaluate a multi-dimensional sum over a product of lower-dimensional local functions. As we will show, for a 2-occ-CNF, the local functions in the formulation have at most two dependencies. This property is precisely what enables a more efficient evaluation via Split-and-List.

\subsection{From Inclusion-Exclusion Formulation to Sum-Product Formulation}
We begin by recalling the inclusion-exclusion formulation for \#SAT~\cite{journals/siamcomp/Iwama89,journals/ipl/Lozinskii92}.
For a CNF formula $\formula=\{C_1,C_2,\dots,C_m\}$ over $n$ variables, let $A_i$ be the set of assignments that \emph{falsify} clause $C_i$. The number of satisfying assignments is given by $\#(\formula) = 2^n - |\bigcup_{i=1}^{m}A_i|$. 
By the inclusion-exclusion principle, the cardinality of this union can be explicitly expanded as:
$\left|\bigcup_{i=1}^{m}A_i\right| = \sum_{i=1}^{m}|A_i| - \sum_{1\le i_1 < i_2 \le m} |A_{i_1}\cap A_{i_2}| +  \cdots + (-1)^{m-1}\left|A_1\cap A_2\cap \cdots \cap A_m\right|$.

To express this more compactly, let $N(S)$ denote the number of assignments that simultaneously falsify every clause in a given subset $S \subseteq \formula$. We can then rewrite the formulation as a summation over all subsets of clauses:
\begin{equation} \label{eq:ie-formulation}
    \#(\formula) = \sum_{S \subseteq \formula} (-1)^{|S|} N(S) \text{ where } N(S) := \begin{cases} 
    0 & \text{if } \exists x\in \var(\formula), \{x, \nl{x}\} \subseteq \ltr(S) \\
    2^{n-\abs{\ltr(S)}} & \text{otherwise}
    \end{cases}.
\end{equation}

We then translate Eq.~\eqref{eq:ie-formulation} to an equivalent sum-product formulation.
To this end, we introduce a vector of Boolean indicator variables $y := (y_1, y_2, \dots, y_m) \in \{0, 1\}^m$, where its $j$-th bit $y_j$ indicates whether the clause $C_j$ is included in the subset $S$ (i.e., $y_j = 1 \iff C_j \in S$). The sign term $(-1)^{|S|}$ naturally factorizes as $\prod_{j=1}^m (-1)^{y_j}$. 

Crucially, the term $N(S)$ factorizes into the product of individual contributions from the $n$ variables. Recall that every assignment in $N(S)$ should simultaneously falsify all clauses in $S$, i.e., setting all literals in $\ltr(S)$ to false. Thus, for a variable $x_i$, the number of valid Boolean assignments to meet this condition depends only on how it appears in the selected clauses $S$, independent of all other variables. Specifically, its contribution is $0$ if both its positive and negative literals appear in $S$ (a logical conflict), $1$ if exactly one positive literal or exactly one negative literal appears (fixing its truth value), and $2$ if it does not appear in $S$ at all (leaving it free). 

We introduce some notations before formally presenting the formulation. Let $\formula = \{C_1, C_2, \dots, C_m\}$ be a $2$-occ-CNF over $n$ variables $V = \{x_1, x_2, \dots, x_n\}$. Recall that for a variable $x_i\in V$, $\deg_{\formula}(x_i)$ denotes the number of times it appears in $\formula$. We assume without loss of generality that $\deg_{\formula}(x_i) \in \{1, 2\}$ for all $x_i$, as variables with $\deg_{\formula}(x_i)=0$ can be trivially factored out. We partition the variable set $V$ into two disjoint subsets:
\begin{itemize}
    \item $V_1 := \{x_i \in V \mid \deg_{\formula}(x_i) = 1\}$: the set of variables appearing exactly once. For each $x_i \in V_1$, let $c_1(i)$ denote the index of the unique clause containing $x_i$.
    \item $V_2 := \{x_i \in V \mid \deg_{\formula}(x_i) = 2\}$: the set of variables appearing exactly twice. For each $x_i \in V_2$, let $c_1(i)$ and $c_2(i)$ denote the indices of the two clauses containing $x_i$, respectively.
\end{itemize}
We say a variable $x_i\in V$ is \emph{mixed} if it has both positive and negative literals in $\formula$; otherwise, $x_i$ is \emph{pure}. 
The sum-product formulation is as follows.

\begin{lemma}[sum-product formulation] \label{lem:sum_product}
    For any 2-occ-CNF $\formula$, it holds that
    \begin{equation} \label{eq:sum_product}
        \#(\formula) = \sum_{y \in \{0, 1\}^m} \left( \prod_{j=1}^m (-1)^{y_j} \prod_{x_i \in V_1} U_i(y_{c_1(i)}) \prod_{x_i \in V_2} W_i(y_{c_1(i)}, y_{c_2(i)}) \right),
    \end{equation}
    where $U_i: \{0, 1\} \to \{1, 2\}$ and $W_i: \{0, 1\}^2 \to \{0, 1, 2\}$ are local weight functions mapping the selection status of the clauses containing variable $x_i$ to the number of valid Boolean assignments of $x_i$.
    Specifically, these functions are defined as follows:
    \begin{itemize}
        \item For any $x_i \in V_1$, $U_i(y_{c_1(i)}) := 2 - y_{c_1(i)}$.
        \item For any $x_i \in V_2$, 
        $W_i(y_{c_1(i)}, y_{c_2(i)}) := 
        \begin{cases}
            2 - \max(y_{c_1(i)}, y_{c_2(i)}) & \text{if } x_i \text{ is pure}, \\
            2 - (y_{c_1(i)} + y_{c_2(i)}) & \text{if } x_i \text{ is mixed}.
        \end{cases}$
    \end{itemize}
\end{lemma}

\begin{proof}[Proof of \Cref{lem:sum_product}]
    By substituting the indicator vector $y \in \{0, 1\}^m$ into Eq.~\eqref{eq:ie-formulation}, the summation over subsets $S$ transforms into a summation over $y$, and the term $(-1)^{|S|}$ transforms to $\prod_{j=1}^m (-1)^{y_j}$. As established, $N(S)$ factorizes into independent valid assignment counts for each variable. It suffices to verify that $U_i$ and $W_i$ correctly compute these counts for variables in $V_1$ and $V_2$, respectively.
    
    For a variable $x_i \in V_1$, the unique clause containing $x_i$ is selected if and only if $y_{c_1(i)} = 1$. If $y_{c_1(i)} = 0$, $x_i$ is free and has $2$ valid assignments; if $y_{c_1(i)} = 1$, its truth value is fixed to falsify the clause and has $1$ valid assignment. The expression $2 - y_{c_1(i)}$ matches these counts.
    
    For a variable $x_i \in V_2$, let $b_1 := y_{c_1(i)}$ and $b_2 := y_{c_2(i)}$ denote the Boolean selection indicators of its two clauses. We verify the soundness of $W_i(b_1, b_2)$ for each type of $x_i$:
    \begin{itemize}
        \item {$x_i$ is pure.} The two literals of $x_i$ share the same polarity. Then, $x_i$ has $1$ valid assignment if at least one of its clauses is selected (i.e., $\max(b_1, b_2) = 1$), and has $2$ valid assignments only if neither is selected (i.e., $b_1=b_2=0$). The expression $2 - \max(b_1, b_2)$ matches this.
        
        \item {$x_i$ is mixed.} The two literals of $x_i$ have opposite polarities. Selecting both clauses ($b_1=b_2=1$) leads to a conflict, leaving $0$ valid assignments. 
        Selecting exactly one clause ($b_1+b_2=1$) leaves $1$ valid assignment, and selecting neither ($b_1=b_2=0$) leaves $2$ valid assignments. The expression $2 - (b_1 + b_2)$ matches these counts.
    \end{itemize}
    Taking the product of these local functions over all variables reconstructs $N(S)$, which completes the proof.
\end{proof}

\subsection{Fast Solving the Sum-Product Problem}
To evaluate our sum-product formulation, we employ the {Split and List} method, whose foundational idea was first introduced by Williams~\cite{journals/tcs/Williams05}. This idea was subsequently adapted by Koivisto~\cite{journals/ipl/Koivisto06} to solve a class of sum-product problems. For completeness and clarity, we provide a concise, self-contained proof that applies Koivisto's approach in \cite{journals/ipl/Koivisto06} to our specific formulation. At a high level, we partition the $m$ clauses into three equal groups and list the $2^{m/3}$ indicator states for each. Because each variable in a $2$-occ-CNF appears in at most two clauses, its local weight depends on at most two of these groups. This allows us to factorize the inner product into three pairwise components and reduce the entire multidimensional summation to computing the trace of the product of three matrices.

\begin{lemma}\label{lem:sum-product-compt}
    The sum-product formulation given in Eq.~\eqref{eq:sum_product} can be evaluated in $O^*(2^{\omega m / 3})$ time and $O^*(2^{2m / 3})$ space.
\end{lemma}

\begin{proof}[Proof of \Cref{lem:sum-product-compt}]
    For convenience, we extend unary local weight functions to binary. For each variable $x_i \in V_1$, let $c_2(i) := c_1(i)$ and $W_i(y_{c_1(i)}, y_{c_2(i)}) := U_i(y_{c_1(i)})$. 
    Eq.~\eqref{eq:sum_product} becomes:
    \begin{equation} \label{eq:unified_sp}
        \#(\formula) = \sum_{y \in \{0, 1\}^m} \left( \prod_{j=1}^m (-1)^{y_j} \prod_{i=1}^n W_i(y_{c_1(i)}, y_{c_2(i)}) \right).
    \end{equation}

    Assume without loss of generality that $m$ is divisible by $3$ (padding with at most two unit clauses on fresh variables if necessary). We partition the clause indices $\{1, \dots, m\}$ into three disjoint subsets $I, J$, and $K$, each of size $m/3$. For any vector $y \in \{0, 1\}^m$, let $y_I, y_J, y_K \in \{0, 1\}^{m/3}$ denote its projections onto the coordinates in $I, J$, and $K$, respectively. 

    Let $E_i := \{c_1(i), c_2(i)\} \subseteq \{1, \dots, m\}$ denote the set of clause indices associated with $x_i$. Since $|E_i| \le 2$ for every $x_i \in V$, no set $E_i$ intersects all three subsets $I, J$, and $K$. This bounded dependency allows us to partition the variable set $V$ as follows:
    \begin{itemize}
        \item $X_{I,J} := \{x_i \in V \mid E_i \subseteq I \cup J\}$,
        \item $X_{J,K} := \{x_i \in V \mid E_i \subseteq J \cup K\} \setminus X_{I,J}$,
        \item $X_{K,I} := \{x_i \in V \mid E_i \subseteq K \cup I\} \setminus (X_{I,J} \cup X_{J,K})$.
    \end{itemize}
    For each $x_i\in X_{P,Q}$, where $(P,Q)\in\{(I,J),(J,K),(K,I)\}$, both clause indices $c_1(i)$ and $c_2(i)$ lie in $P\cup Q$. Consequently, its local weight $W_i(y_{c_1(i)},y_{c_2(i)})$ depends only on the sub-vectors $y_P$ and $y_Q$.
    
    We distribute the terms in Eq.~\eqref{eq:unified_sp} into three pairwise functions:
    \begin{align*}
        F_{I,J}(y_I, y_J) &:= \left( \prod_{j \in I} (-1)^{y_j} \right) \prod_{x_i \in X_{I,J}} W_i(y_{c_1(i)}, y_{c_2(i)}), \\
        F_{J,K}(y_J, y_K) &:= \left( \prod_{j \in J} (-1)^{y_j} \right) \prod_{x_i \in X_{J,K}} W_i(y_{c_1(i)}, y_{c_2(i)}), \\
        F_{K,I}(y_K, y_I) &:= \left( \prod_{j \in K} (-1)^{y_j} \right) \prod_{x_i \in X_{K,I}} W_i(y_{c_1(i)}, y_{c_2(i)}). 
    \end{align*}
    Multiplying these exactly reconstructs the summand. Thus, the formulation factorizes as:
    \begin{align} \label{eq:factorized}
        \#(\formula) 
        &= \sum_{y \in \{0, 1\}^m} \left( \prod_{j=1}^m (-1)^{y_j} \prod_{i=1}^n W_i(y_{c_1(i)}, y_{c_2(i)}) \right) \nonumber \\
        &= \sum_{y_I} \sum_{y_J} \sum_{y_K} F_{I,J}(y_I, y_J) \cdot F_{J,K}(y_J, y_K) \cdot F_{K,I}(y_K, y_I).
    \end{align}

    We define three $2^{m/3} \times 2^{m/3}$ matrices $A, B$, and $C$ with entries $A[y_I, y_J] := F_{I,J}(y_I, y_J)$, $B[y_J, y_K] := F_{J,K}(y_J, y_K)$, and $C[y_K, y_I] := F_{K,I}(y_K, y_I)$. Recall that the trace of a matrix $M$, denoted by $\textrm{Tr}(M)$, is the sum of its diagonal entries. By the definition of matrix multiplication, Eq.~\eqref{eq:factorized} equals the trace of the product of these three matrices:
    \begin{equation}
        \#(\formula) = \sum_{y_I} (A B C)[y_I, y_I] = \textrm{Tr}(A B C).
    \end{equation}

    We then analyze the time and space required by the above approach. Generating the matrices requires evaluating the pairwise functions for all $2^{2m/3}$ sub-vector pairs, taking $\bigOstar{2^{2m/3}}$ time and space. Computing the product $ABC$ via fast matrix multiplication takes $\bigOstar{(2^{m/3})^\omega} = \bigOstar{2^{m \omega / 3}}$ time and $\bigOstar{2^{2m/3}}$ space, which completes the proof.
\end{proof}

Combining \Cref{lem:sum_product} and \Cref{lem:sum-product-compt}, we have:
\TwoOccCNFComplexityThm*

\section{Generalizing to Weighted Model Counting}
In this section, we demonstrate that our algorithms for \#$k$-SAT and \#$2$-occ-SAT can be naturally extended to the weighted setting with the same time complexity.
\begin{definition}[Weighted Model Count]
    Let $\ltr(\formula):=\bigcup_{x\in\var(\formula)}\{x,\nl{x}\}$ be the set of literals over $\var(\formula)$.
    Let $\wfunc:\ltr(\formula)\rightarrow \mathbb{Z}$ be a \emph{weight function} that assigns an integer weight (possibly negative) to each literal. Let $\mathcal{M}(\formula)$ be the set of all satisfying assignments of $\formula$.
    The \emph{weighted model count} $\WMC(\formula, \wfunc)$ is defined as  
    \[
         \WMC(\formula, \wfunc):=
         \sum_{{\model\in \mathcal{M}(\formula)}}
         \Big(
         \quad~
         \prod_{\mathclap{\substack{x\in \var(\formula)\\ \model(x)=1}}}\wfunc(x)
         \cdot
         \prod_{\mathclap{\substack{y\in \var(\formula)\\ \model(y)=0}}}\wfunc(\nl{y})
         \Big)
         .
    \]
\end{definition}
In the {\sc Weighted Model Counting} (WMC) problem, given a formula $\formula$ and a weight function $\wfunc$, the goal is to compute the weighted model count $\WMC(\formula, \wfunc)$ of formula $\formula$.
We denote by weighted \#$k$-SAT and weighted \#$2$-occ-SAT the restricted versions of WMC where the inputs are $k$-CNF and $2$-occ-CNF formulas, respectively.

Before detailing the specific extensions, we analyze the computational cost of handling weights.
Let
$W_{\max}:=\max\left\{1,\max_{\ell\in\ltr(\formula)}|\wfunc(\ell)|\right\}$.
Every integer arising in the branching recursion or the sum-product computation has absolute value at most
$2^{\bigO{n+m}}W_{\max}^{\bigO{n}}$ and therefore has bit length
$\bigO{(m+n)\log(1+W_{\max})\bigr}$.
Thus, arithmetic operations on these numbers incur only a polynomial overhead, which is absorbed by the $\bigOstar{\cdot}$ notation.

\subsection{Weighted \#$k$-SAT}
The algorithm \algcount{} in \Cref{subsec:branching} can be modified for the weighted setting. For the base case where $\formula = \emptyset$, the algorithm returns $\prod_{x \in \var(\formula)} (\wfunc(x) + \wfunc(\nl{x}))$, which is the total weight of all possible assignments to the remaining variables. Tautology reduction applies unchanged in the weighted setting.
For a clause $C \in \formula$, the clause branching rule (\Cref{lem:clause-branching}) generalizes to the weighted setting as follows:
\[
    \WMC(\formula, \wfunc) = \WMC(\formula \setminus \{C\}, \wfunc) - \WMC((\formula \setminus \{C\})[C=0], \wfunc)\cdot \prod_{\ell\in C}\wfunc(\nl{\ell}).
\]
Similarly, for an isolated clause $C$, the reduction rule (\Cref{lem:isolated-clause}) generalizes to:
\[
    \WMC(\formula, \wfunc) = \WMC(\operatorname{rem}(\formula,C), \wfunc) \cdot \Big(\prod_{\ell \in C}(\wfunc(\ell) + \wfunc(\nl{\ell})) - \prod_{\ell\in C}\wfunc(\nl{\ell})\Big).
\]
The correctness of these weighted extensions follows the same logic as the proofs of \Cref{lem:clause-branching} and \Cref{lem:isolated-clause}. 
For the branching rule, we partition
$\mathcal{M}(\formula\setminus\{C\})$ into assignments satisfying $C$ and those falsifying $C$. The former are precisely $\mathcal{M}(\formula)$. The latter are in bijection with $\mathcal{M}((\formula\setminus\{C\})[C=0])$: each residual assignment extends uniquely by setting every literal in $C$ to false. The weight of this extension is the weight of the residual assignment multiplied by $\prod_{\ell\in C}\wfunc(\nl{\ell})$. Since the total weight is additive over disjoint sets of assignments, the subtraction correctly yields the total weight of $\mathcal{M}(\formula)$. 
The isolated clause reduction holds because the weights of independent variables are multiplicative; the term $\prod_{\ell \in C}(\wfunc(\ell) + \wfunc(\nl{\ell}))$ represents the total weight of all assignments to the variables in $C$, while $\prod_{\ell \in C}\wfunc(\nl{\ell})$ is the weight of the unique assignment that falsifies $C$.

Since the measure drop in the branching algorithm is the same as in the unweighted case, the complexity is the same as that in \Cref{thm:k-cnf-complexity}.

\begin{theorem}\label{thm:weighted-k-sat}
    Weighted \#$k$-SAT with $m$ clauses can be solved in $\bigOstar{2^{m(1 - \frac{1}{5k})}} \subseteq \bigOstar{2^{m(1 - \Omega(1/k))}}$ time and polynomial space.
\end{theorem}

\subsection{Weighted \#$2$-occ-SAT}
The inclusion-exclusion formulation extends to the weighted setting: For a formula $\formula$ and weight function $\wfunc$, let $N_{\wfunc}(S)$ denote the total weight of assignments falsifying every clause in $S \subseteq \formula$. Then it holds that $\WMC(\formula, \wfunc) = \sum_{S \subseteq \formula} (-1)^{|S|} N_{\wfunc}(S)$.
As in \Cref{sec:2-occ-CNF}, variables of degree zero can be factored out independently; in the weighted setting, each such variable $x$ contributes a factor $\wfunc(x)+\wfunc(\nl{x})$.
We therefore assume that every variable belongs to $V_1\cup V_2$.
Analogously to \Cref{lem:sum_product}, $N_{\wfunc}(S)$ then factorizes into the independent contributions of the variables.

\begin{lemma}[weighted sum-product formulation] \label{lem:weighted_sum_product}
    For any 2-occ-CNF $\formula$ and weight function $\wfunc$, it holds that:
    \begin{equation*}
        \WMC(\formula, \wfunc) = \sum_{y \in \{0, 1\}^m} \left( \prod_{j=1}^m (-1)^{y_j} \prod_{x_i \in V_1} \tilde{U}_i(y_{c_1(i)}) \prod_{x_i \in V_2} \tilde{W}_i(y_{c_1(i)}, y_{c_2(i)}) \right),
    \end{equation*}
    where the local functions $\tilde{U}_i: \{0, 1\} \to \mathbb{Z}$ and $\tilde{W}_i: \{0, 1\}^2 \to \mathbb{Z}$ are defined as follows:
    \begin{itemize}
        \item For $x_i \in V_1$, let $\ell$ be the literal of $x_i$ appearing in $C_{c_1(i)}$. \\
        Then $\tilde{U}_i(y_{c_1(i)}) := (\wfunc(x_i) + \wfunc(\nl{x_i})) \cdot (1 - y_{c_1(i)}) + \wfunc(\nl{\ell}) \cdot y_{c_1(i)}$.
        
        \item For a pure variable $x_i \in V_2$, let $\ell$ be the literal of $x_i$ appearing in both $C_{c_1(i)}$ and $C_{c_2(i)}$. \\
        Then $\tilde{W}_i(y_{c_1(i)}, y_{c_2(i)}) := (\wfunc(x_i) + \wfunc(\nl{x_i})) \cdot (1 - \max(y_{c_1(i)}, y_{c_2(i)})) + \wfunc(\nl{\ell}) \cdot \max(y_{c_1(i)}, y_{c_2(i)})$.
        
        \item For a mixed variable $x_i \in V_2$, we assume without loss of generality that $C_{c_1(i)}$ contains the positive literal $x_i$ and $C_{c_2(i)}$ contains the negative literal $\nl{x_i}$. 
        Then $\tilde{W}_i(y_{c_1(i)}, y_{c_2(i)}) := (\wfunc(x_i) + \wfunc(\nl{x_i}))(1-y_{c_1(i)})(1-y_{c_2(i)}) + \wfunc(\nl{x_i})y_{c_1(i)}(1-y_{c_2(i)}) + \wfunc(x_i)(1-y_{c_1(i)})y_{c_2(i)}$.
    \end{itemize}
\end{lemma}

The above formulation has the same sum-product structure as its unweighted counterpart.
It can therefore be evaluated by the same Split-and-List procedure using fast matrix multiplication, as described in \Cref{lem:sum-product-compt}. Consequently, the overall complexity is the same as that in \Cref{thm:2-occ-cnf-complexity}.

\begin{theorem}\label{thm:weighted-2-occ-sat}
    Weighted \#$2$-occ-SAT with $m$ clauses can be solved in $\bigOstar{2^{m \omega / 3}}\subseteq \bigOstar{1.7297^m}$ time and exponential space, where $\omega < 2.3714$ denotes the matrix multiplication exponent.
\end{theorem}

\section{Hardness of \#$d$-occ-SAT}\label{sec:max-to-occ}
In this section, we prove \Cref{thm:max-to-occ}, restated below.

\MaxToOccReductionThm*

Our proof uses polynomial interpolation, a technique for reductions between counting problems.
It proceeds in three steps: first, encode the desired information in the coefficients of a polynomial of polynomially bounded degree; second, construct counting instances whose counts give the values of this polynomial at sufficiently many distinct points; and third, recover the polynomial by interpolation and decode the desired information from its coefficients.
This paradigm has been used in counting-complexity reductions~\cite{Linial,JVW,Vadhan,journals/talg/DellHMTW14,journals/iandc/Curticapean18}.

Following this paradigm, our reduction has the following outline. Suppose $\Phi$ is the input \#Max-$d$-SAT formula.
\begin{itemize}
\item Construct a polynomial $P_{\Phi}$ whose coefficients encode the optimum and the number of optimal assignments.
\item With the aim of constructing $P_{\Phi}$, we first express assignment scores for \#Max-$d$-SAT as a sum of subset indicators. Similarly, we express the inclusion-exclusion formulation for \#$d$-occ-SAT on positive formulas in terms of subset indicators. This common structure guides the construction of $P_{\Phi}$ as well as the construction of \#$d$-occ-SAT instances.
\item Construct \#$d$-occ-SAT instances on positive formulas to realize evaluations of this polynomial.
Each instance has one clause per variable of the \#Max-$d$-SAT instance.
\item Recover the polynomial by Lagrange interpolation and obtain the optimum and the number of optimal assignments from its coefficients.
\end{itemize}

Throughout this section, $d\ge2$ is fixed.

\subsection{Two formulations over subsets}
For every nonnegative integer $k$, write $[k]=\{1,\ldots,k\}$, with $[0]=\varnothing$.
For a proposition $\mathcal P$, let $\ind{\mathcal P}$ be $1$ if it holds and $0$ otherwise.
All empty products equal $1$.

\subsubsection{The formulation for \#Max-$d$-SAT}
Let $\Phi=C_1\wedge\cdots\wedge C_m$ be the $d$-CNF of the \#Max-$d$-SAT instance, over the variable set $\var(\Phi)=\{x_1,\ldots,x_n\}$.
An assignment is represented by $S=\{i\in[n]:x_i=1\}$.
Its score is the number of clauses it satisfies.
Each copy of an identical clause contributes separately to the score.
Without loss of generality, we assume that $\Phi$ contains no tautological clauses.

For a clause $C$, define $C^+=\{i:x_i\in C\}$ and $C^-=\{i:\nl{x_i}\in C\}$.
These sets are disjoint.
Let $\score_C(S)$ be $1$ if the assignment $S$ satisfies $C$ and $0$ otherwise.
A clause is falsified exactly when all its literals are false.
Expanding gives
\begin{align}
\score_C(S)
 &=1-\prod_{i\in C^+}(1-\ind{i\in S})\prod_{i\in C^-}\ind{i\in S}\notag\\
 &=1-\sum_{U\subseteq C^+}(-1)^{|U|}\ind{C^-\cup U\subseteq S}.
 \label{red-eq:clause}
\end{align}
For example, for $C=x_1\vee\nl{x_2}$, $\score_C(S)=1-\ind{\{2\}\subseteq S}+\ind{\{1,2\}\subseteq S}$.
The score is $0$ exactly when $S=\{2\}$, i.e., $x_1=0$ and $x_2=1$.

Let $\score_\Phi(S)$ be the total number of clauses satisfied by $S$, and let $\OPT(\Phi)=\max_{S\subseteq[n]}\score_\Phi(S)$.
Every subset $C^-\cup U$ has size at most $d$.
Summing \eqref{red-eq:clause} and collecting identical subsets therefore produces integer coefficients $a_A$ with
\begin{equation}\label{red-eq:score}
\score_\Phi(S)=\sum_{\substack{A\subseteq[n]\\|A|\leq d}}
a_A\ind{A\subseteq S}.
\end{equation}
The coefficients can be computed by expanding at most $2^d$ terms per clause.

\subsubsection{The formulation for \#$d$-occ-SAT on positive formulas}
Let $\Psi=D_1\wedge\cdots\wedge D_M$ be a positive $d$-occ-CNF over $\var(\Psi)$.
For each $y\in\var(\Psi)$, define
\begin{equation}\label{red-eq:occ}
\Occ_y=\{i\in[M]:y\in D_i\}.
\end{equation}
The occurrence bound gives $|\Occ_y|\leq d$.
For $R\subseteq[M]$, falsifying the clauses indexed by $R$ sets all variables in $\bigcup_{i\in R}D_i$ to zero and leaves the others free.
Inclusion-exclusion gives
\begin{align}
\#(\Psi)
 &=\sum_{R\subseteq[M]}(-1)^{|R|}
   \bigl|\{\model:\model\text{ falsifies every }D_i,\ i\in R\}\bigr|\notag\\
 &=\sum_{R\subseteq[M]}(-1)^{|R|}
   2^{|\var(\Psi)|-|\bigcup_{i\in R}D_i|}\notag\\
 &=\sum_{R\subseteq[M]}(-1)^{|R|}
   2^{|\{y\in\var(\Psi):\Occ_y\cap R=\varnothing\}|}\notag\\
 &=(-1)^M\sum_{S\subseteq[M]}(-1)^{|S|}
   2^{\sum_{y\in\var(\Psi)}\ind{\Occ_y\subseteq S}}.
 \label{red-eq:count}
\end{align}
For the last equality, we substitute $S=[M]\setminus R$.
Under this substitution, $\Occ_y\cap R=\varnothing$ becomes $\Occ_y\subseteq S$, and $|R|=M-|S|$ gives $(-1)^{|R|}=(-1)^M(-1)^{|S|}$.

\subsection{The Reduction}
With the formulations above, we are ready to present the reduction in \Cref{thm:max-to-occ}.
The main idea is as follows.
Let $\Phi$ be the input $d$-CNF on $n$ variables for \#Max-$d$-SAT.
Our aim is to construct a univariate polynomial $P_\Phi$ whose degree is polynomially bounded in the size of $\Phi$ and whose coefficients allow us to compute both $\OPT(\Phi)$ and the number of assignments attaining this score.
We construct $P_\Phi$ so that, for each integer $r\geq1$, we can build a positive $d$-occ-CNF $\Psi_r$ with exactly $n$ clauses satisfying $\#(\Psi_r)=(-1)^nP_\Phi(2^r)$.
The number of satisfying assignments for $\Psi_r$ therefore allows us to compute $P_\Phi(2^r)$.
Then, by choosing polynomially many distinct values of $r$, we obtain sufficiently many evaluations to recover the coefficients of $P_\Phi$ using Lagrange interpolation.
Finally, we use these coefficients to compute $\OPT(\Phi)$ and the number of optimal assignments.
Below, we first present the construction of $P_\Phi$ and then the construction of positive $d$-occ-CNF $\Psi_r$.

\subsubsection{Encoding the objective in a polynomial}
For clarity, before defining the univariate polynomial $P_\Phi$, we introduce an auxiliary multivariate polynomial $Q_\Phi$:
\begin{equation}\label{red-eq:Q}
Q_\Phi(X,Y,Z)=\sum_{S\subseteq[n]}(-1)^{|S|}X^{p(S)}Y^{q(S)}Z^{|S|},
\end{equation}
where $p(S)$ and $q(S)$ are the sums of contributions from the positive and negative coefficients $a_A$ with $A\ne\varnothing$ in \eqref{red-eq:score}, respectively, with the negative coefficients taken in absolute value:
\begin{align}
p(S)&=\sum_{\substack{A\ne\varnothing\\a_A>0}}a_A\ind{A\subseteq S},
&q(S)&=\sum_{\substack{A\ne\varnothing\\a_A<0}}(-a_A)\ind{A\subseteq S}.
\label{red-eq:pq}
\end{align}
Here and below, sums over $A$ range over subsets of $[n]$ of size at most $d$.
Setting $c=a_\varnothing$, we have
\begin{equation}
\score_\Phi(S)=c+p(S)-q(S).
\end{equation}

We now give intuition behind the definition of $Q_\Phi$.
We want to evaluate a polynomial using the count of a \#$d$-occ-SAT instance on a positive formula.
A natural starting point is the score generating function $F_\Phi(X)=\sum_{S\subseteq[n]}X^{\score_\Phi(S)}$ (see, e.g.,~\cite{journals/talg/ScottS09}).
Its largest exponent is the optimum, and the corresponding coefficient is the number of optimal assignments.
However, using it here faces two obstacles.

First, the formulation for \#$d$-occ-SAT on positive formulas in \eqref{red-eq:count} contains the factor $(-1)^{|S|}$.
Incorporating this factor directly into $F_\Phi$ gives $\widetilde F_\Phi(X)=\sum_{S\subseteq[n]}(-1)^{|S|}X^{\score_\Phi(S)}$.
However, equal scores can receive opposite signs and cancel: if exactly two assignments have score $j$, one with even $|S|$ and one with odd $|S|$, the coefficient of $X^j$ is $1-1=0$.
To prevent this cancellation, we can record $|S|$ with an additional variable and use the modified polynomial
\[
G_\Phi(X,Z)=\sum_{S\subseteq[n]}(-1)^{|S|}X^{\score_\Phi(S)}Z^{|S|}.
\]
All contributions to the coefficient of $X^jZ^h$ now have the same sign $(-1)^h$, so they cannot cancel.

Second, with $M=n$, the exponent in \eqref{red-eq:count} can be expressed as $g(S)=\sum_A k_A\ind{A\subseteq S}$, where $k_A\geq0$ counts variables with occurrence set $A$.
However, the coefficients $a_A$ in \eqref{red-eq:score} may be negative, whereas each $k_A$ must be nonnegative because it counts variables.
Thus we record the contributions from positive and negative coefficients, represented by $p(S)$ and $q(S)$, in separate coordinates, yielding $Q_\Phi$ in \eqref{red-eq:Q}, and subtract only when decoding $c+p(S)-q(S)$.

We then use $Q_\Phi$ to construct the desired univariate polynomial $P_\Phi$.
To this end, define
\begin{equation}\label{red-eq:base}
W=\sum_{A\ne\varnothing}|a_A|,\qquad B=W+1.
\end{equation}
Thus $0\leq p(S),q(S)<B$.
Expanding each clause gives terms whose absolute coefficient values sum to at most $1+2^d$.
Hence
\begin{equation}\label{red-eq:Bbound}
B\leq 1+m(1+2^d).
\end{equation}
In particular, $B=O(m)$ for fixed $d$.
Define the new univariate polynomial
\begin{equation}\label{red-eq:P}
P_\Phi(T)=Q_\Phi(T,T^B,T^{B^2})
=\sum_{S\subseteq[n]}(-1)^{|S|}T^{E(S)},
\text{ where }
E(S)=p(S)+Bq(S)+B^2|S|.
\end{equation}
Since $0\leq p(S),q(S)<B$ for every assignment $S$, the degree of $P_\Phi$ is at most
\begin{equation}\label{red-eq:degree}
(B-1)+B(B-1)+B^2n=B^2(n+1)-1=:D.
\end{equation}
For fixed $d$, the bound $D$ is polynomial in the input size of the \#Max-$d$-SAT instance.

\begin{lemma}[Coefficient decoding]\label{red-lem:decode}
For every nonnegative integer $e$, let $b_e$ denote the coefficient of $T^e$ in $P_\Phi(T)$, and set $p_e=e\bmod B$, $q_e=\lfloor e/B\rfloor\bmod B$, and $h_e=\lfloor e/B^2\rfloor$.
Then $b_e=(-1)^{h_e}|\{S\subseteq[n]:E(S)=e\}|$ for every $e\geq0$.
Define the support $\supp(P_\Phi)=\{e:b_e\ne0\}$, the set of exponents with nonzero coefficients.
Then
\begin{equation}\label{red-eq:optdecode}
\OPT(\Phi)=c+\max_{e\in\supp(P_\Phi)}(p_e-q_e).
\end{equation}
Moreover, the number of optimal assignments is
\begin{equation}\label{red-eq:countdecode}
\sum_{\substack{e\in\supp(P_\Phi)\\c+p_e-q_e=\OPT(\Phi)}}|b_e|.
\end{equation}
\end{lemma}
\begin{proof}
For an assignment with $E(S)=e$, taking $e$ modulo $B$ gives $p(S)=p_e$.
Taking $\lfloor e/B\rfloor$ modulo $B$ gives $q(S)=q_e$.
Finally, $0\leq p(S)+Bq(S)<B^2$, so taking $\lfloor e/B^2\rfloor$ gives $|S|=h_e$.
These statements remain valid for $B=1$: then $p=q=0$ and $E(S)=|S|$.
Thus every assignment contributing to $b_e$ has sign $(-1)^{h_e}$, proving the stated formula for $b_e$.
The coefficient $b_e$ is nonzero exactly when there is an assignment $S$ with $E(S)=e$.
The identity $\score_\Phi(S)=c+p(S)-q(S)$ then gives \eqref{red-eq:optdecode}.
The support is nonempty because there is at least one assignment and no coefficient cancellation.
By this formula, $|b_e|$ is the number of assignments with exponent $e$, each having score $c+p_e-q_e$.
Summing these counts over all $e$ with $c+p_e-q_e=\OPT(\Phi)$ therefore gives the number of optimal assignments, as stated in \eqref{red-eq:countdecode}.
\end{proof}

\subsubsection{Realizing evaluations by \#$d$-occ-SAT instances on positive formulas}
For the given \#Max-$d$-SAT instance $\Phi$ and an integer $r\geq1$, we construct a positive CNF $\Psi_r$ as follows.
Create $n$ clauses $D_1,\ldots,D_n$ and introduce the following three groups of distinct fresh variables:
\begin{itemize}
\item Group 1: For each nonempty $A$ with $a_A>0$, introduce $ra_A$ variables, each placed in exactly the clauses $D_i$ with $i\in A$.
\item Group 2: For each nonempty $A$ with $a_A<0$, introduce $rB(-a_A)$ variables, each placed in exactly the clauses $D_i$ with $i\in A$.
\item Group 3: For every $i\in[n]$, introduce $rB^2$ variables placed only in $D_i$.
\end{itemize}
All literals are positive.
A variable from groups 1 and 2 has occurrence set $A$, hence occurs in at most $d$ clauses; a variable from group 3 occurs once.
These last variables also ensure that each clause is nonempty when $n\geq1$.
The formula has exactly $n$ clauses.

\begin{lemma}[Evaluation identity]\label{red-lem:evaluation}
For every integer $r\geq1$, $\#(\Psi_r)=(-1)^nP_\Phi(2^r)$.
\end{lemma}
\begin{proof}
Fix $S\subseteq[n]$.
Let $V_1,V_2,V_3$ be the sets of fresh variables introduced in groups 1, 2, and 3, respectively.
They form a disjoint partition $\var(\Psi_r)=V_1\mathbin{\dot\cup}V_2\mathbin{\dot\cup}V_3$.
Each variable contributes $1$ to {the sum $\sum_{y\in \var(\Psi_r)}\ind{\Occ_y\subseteq S}$} if its occurrence set is contained in $S$, and $0$ otherwise.

For each nonempty $A$ with $a_A>0$, $V_1$ contains $ra_A$ variables with $\Occ_y=A$.
For each nonempty $A$ with $a_A<0$, $V_2$ contains $rB(-a_A)$ variables with the same occurrence set.
For each $i\in[n]$, $V_3$ contains $rB^2$ variables with $\Occ_y=\{i\}$.
Therefore, using the definitions of $p(S)$ and $q(S)$ in \eqref{red-eq:pq},
\begin{align*}
\sum_{y\in V_1}\ind{\Occ_y\subseteq S}
&=r\sum_{\substack{A\ne\varnothing\\a_A>0}}a_A\ind{A\subseteq S}=rp(S),\\
\sum_{y\in V_2}\ind{\Occ_y\subseteq S}
&=rB\sum_{\substack{A\ne\varnothing\\a_A<0}}(-a_A)\ind{A\subseteq S}=rBq(S),\\
\sum_{y\in V_3}\ind{\Occ_y\subseteq S}
&=rB^2\sum_{i\in[n]}\ind{\{i\}\subseteq S}=rB^2|S|.
\end{align*}
Summing over the partition gives
\[
\sum_{y\in\var(\Psi_r)}\ind{\Occ_y\subseteq S}
=\sum_{j=1}^3\sum_{y\in V_j}\ind{\Occ_y\subseteq S}
=rp(S)+rBq(S)+rB^2|S|=rE(S).
\]
Substituting into the formulation for \#$d$-occ-SAT on positive formulas in \eqref{red-eq:count} with $M=n$ gives
\[
\#(\Psi_r)=(-1)^n\sum_{S\subseteq[n]}(-1)^{|S|}2^{rE(S)}
=(-1)^nP_\Phi(2^r),
\]
as required.
\end{proof}

\subsubsection{Putting it all together}
We are now ready to prove \Cref{thm:max-to-occ}.
\begin{proof}[Proof of \Cref{thm:max-to-occ}]

Let $\Phi$ be a given \#Max-$d$-SAT instance on $n$ variables.
We compute its number of optimal assignments using \#$d$-occ-SAT queries on positive formulas with exactly $n$ clauses.
With $D$ as defined in \eqref{red-eq:degree}, query $\Psi_r$ for $r=1,\ldots,D+1$.
Using the counts returned by these \#$d$-occ-SAT queries on positive formulas, Lemma~\ref{red-lem:evaluation} gives $P_\Phi(2^r)$ at $D+1$ distinct points.
Since the degree of $P_\Phi$ is at most $D$, we can recover its coefficients from these $D+1$ values by Lagrange interpolation.
Lemma~\ref{red-lem:decode} then gives $\OPT(\Phi)$ and the number of optimal assignments via \eqref{red-eq:optdecode} and \eqref{red-eq:countdecode}, respectively.

We now analyze the time complexity of the reduction.
Each coefficient $a_A$ with $A\ne\varnothing$ contributes at most $rB|a_A|$ variables, so the number $N_r$ of variables in $\Psi_r$ satisfies $N_r\leq r(BW+nB^2)$.
There are $D+1=B^2(n+1)$ queries, and each formula $\Psi_r$ has length at most $dN_r$.
By the bound on $B$ in \eqref{red-eq:Bbound} and the definition of $D$ in \eqref{red-eq:degree}, both the number of queries and the size of each formula are polynomial in the input size.
Moreover, $\#(\Psi_r)\leq2^{N_r}$, so each returned count requires at most $N_r+1$ bits.
Since the degree upper bound $D$, the bit lengths of the points $2^r$, and the bit lengths of the returned counts are all polynomially bounded in the input size, Lagrange interpolation can be performed in polynomial time.
Summing the absolute values of the relevant coefficients takes polynomial time.

All queries depend only on $\Phi$, so the reduction is non-adaptive.
There are polynomially many queries, each with exactly $n$ clauses, and all remaining computation takes polynomial time.
Thus an algorithm for \#$d$-occ-SAT on positive formulas running in $\bigOstar{\gamma^M}$ time yields an algorithm for \#Max-$d$-SAT running in $\bigOstar{\gamma^n}$ time.
If the former algorithm uses polynomial space, so does the latter.
This completes the proof.
\end{proof}

\section{Conclusion and Discussion}
We have explored the complexity landscape of \#SAT parameterized by the number of clauses $m$.
We have presented algorithms that break the $2^m$-barrier for \#$k$-SAT and \#2-occ-SAT, and these results naturally generalize to the weighted setting.
Observe that if one recursively applies the clause branching rule to all clauses, the resulting search tree exactly corresponds to the standard subset-based inclusion-exclusion formulation. Basically, our algorithms exploit the specific structural properties of $k$-CNF and $2$-occ-CNF, respectively, to evaluate this formulation without explicitly exploring all $2^m$ terms.

For \#$d$-occ-SAT with $d\ge3$, our results suggest that breaking the $2^m$-barrier may be difficult: for every fixed $d\ge3$, no $\bigOstar{(2-\varepsilon)^m}$-time algorithm exists for any constant $\varepsilon>0$, even on positive formulas, under the \#Max-3-SAT Hypothesis.
It is worth mentioning that, in contrast, the $2^m$-barrier can be broken when computing only the parity of the number of satisfying assignments of a $d$-occ-CNF~\cite{conf/sat/JainPSTX26}.

\section*{Declarations}
We used OpenAI ChatGPT to assist in proofreading, language polishing, and checking the clarity and logical consistency of the proofs, resulting in minor modifications. In addition, the idea for the reduction in \Cref{thm:max-to-occ} was suggested by ChatGPT-5.6 Sol; the authors checked its correctness and rewrote its exposition. The authors assume responsibility for all content.

\bibliography{_ref}
\bibliographystyle{alpha}
\end{document}